\documentclass[submission,copyright,creativecommons]{eptcs}
\providecommand{\event}{NCMA 2025} 

\usepackage[english]{babel}
\usepackage{iftex}
\usepackage{hyperref}
\usepackage{tikz,pgf}
\usetikzlibrary{positioning, arrows, automata}
\usepackage{comment}
\usepackage{amsmath, amsfonts, amssymb, theorem, enumerate}

\newcommand{\NIL}{\mathit{NIL}}
\newcommand{\COMB}{\mathit{COMB}}
\newcommand{\DEF}{\mathit{DEF}}
\newcommand{\SYDEF}{\mathit{SYDEF}}
\newcommand{\COMM}{\mathit{COMM}}
\newcommand{\CIRC}{\mathit{CIRC}}
\newcommand{\SUF}{\mathit{SUF}}
\newcommand{\PRE}{\mathit{PRE}}
\newcommand{\INF}{\mathit{INF}}
\newcommand{\NC}{\mathit{NC}}
\newcommand{\SF}{\mathit{SF}}
\newcommand{\UF}{\mathit{UF}}
\newcommand{\FIN}{\mathit{FIN}}
\newcommand{\MON}{\mathit{MON}}
\newcommand{\PS}{\mathit{PS}}
\newcommand{\ORD}{\mathit{ORD}}
\newcommand{\STAR}{\mathit{STAR}}
\newcommand{\COM}{\mathit{COM}}
\newcommand{\LCOM}{\mathit{LCOM}}
\newcommand{\RCOM}{\mathit{RCOM}}
\newcommand{\TCOM}{\mathit{2COM}}
\newcommand{\REG}{\mathit{REG}}
\newcommand{\Suf}{\mathit{Suf}}
\newcommand{\Inf}{\mathit{Inf}}
\newcommand{\Pre}{\mathit{Pre}}
\newcommand{\ec}[1]{\ensuremath{\mathcal{EC}(#1)}}
\newcommand{\ic}[1]{\ensuremath{\mathcal{IC}(#1)}}
\newcommand{\ellA}{\ell_{\mathrm{A}}}
\newcommand{\ellC}{\ell_{\mathrm{C}}}
\newcommand{\qmand}{\quad\mbox{and}\quad}

\newcommand{\iirule}[4]{
\begin{align*}
  S_1 &= #1, & S_2 &= #3,\\
  C_1 &= #2, & C_2 &= #4
\end{align*}
}

\def\Set#1#2{\left\{\: #1\;|\; #2\:\right\}}

\def\Sets#1{\left\{\,#1\,\right\}}
\def\set#1#2{\{\; #1 \mid #2\;\}}
\def\sets#1{\{#1\}}

\def\cF{{\cal F}}

\def\cS{{\cal S}}

\def\cEC{{\cal EC}}
\def\cIC{{\cal IC}}
\DeclareSymbolFont{symbols}{OMS}{cmsy}{m}{n}

\def\Lra{\Longrightarrow}

\newtheorem{theorem}{Theorem}
\newtheorem{lemma}[theorem]{Lemma}

\theorembodyfont{\rm}
\newtheorem{example}[theorem]{Example}
\newenvironment{proof}{{\em Proof. }}{{}\hspace*{\fill}$\Box$ \par \medskip }
\newenvironment{proof*}{{\em Proof. }}{\par \medskip }

\newlength{\btlabelwidth}
\newlength{\btleftmargin}
\newenvironment{btlists}{\begin{list}{{\rm--}}{%
\setlength{\labelwidth}{\btlabelwidth}\setlength{\leftmargin}{\btleftmargin}%
\setlength{\topsep}{0pt plus0.2ex}%
\setlength{\itemsep}{0ex plus0.2ex}%
\setlength{\parsep}{0pt plus0.2ex}}}{\end{list}}
\tikzstyle{to}=[->, >=stealth]
\tikzstyle{hier}=[->, >=angle 60]
\tikzstyle{hiero}=[->, >=angle 60, dashed]
\tikzstyle{state}=[circle,draw,inner sep=2pt,minimum size=8mm]
\tikzstyle{edgeLabel}=[inner sep=0.5mm,fill=white,text=black]

\usepackage{pdflscape}
\ifpdf
  \usepackage{underscore}     
  \usepackage[T1]{fontenc}    
\else
  \usepackage{breakurl}       
\fi

\title{Idefix-Closed Languages and Their Application\\ in Contextual Grammars}
\author{Marvin K\"odding
\institute{Institut f{\"u}r Mathematik und Informatik, P{\"a}dagogische Hochschule Heidelberg\\Im Neuenheimer Feld 561, 69120 Heidelberg, Germany}
\email{koedding@ph-heidelberg.de}
\and
Bianca Truthe
\institute{Institut f\"ur Informatik, Universit\"at Giessen\\Arndtstr. 2, 35392 Giessen, Germany}
\email{bianca.truthe@informatik.uni-giessen.de}
}
\def\titlerunning{Idefix-Closed Languages and Their Application in Contextual Grammars}
\def\authorrunning{M. K\"odding \& B. Truthe}

\begin{document}
\maketitle

\begin{abstract}
In this paper, we continue the research on the power of contextual grammars with selection languages from 
subfamilies of the family of regular languages. 
We investigate infix-, prefix-, and suffix-closed languages (referred to as idefix-closed languages) and compare
such language families to some other subregular families of languages (finite, monoidal, nilpotent, combinational, 
(symmetric) definite, ordered, non-counting,  power-sepa\-rating, commutative, circular, union-free, star, and comet languages). Further, we compare the families of the hierarchies obtained for external and internal
contextual grammars with the language families defined by these new types for the selection. In this way, we extend the existing hierarchies by new language families.
Moreover, we solve an open problem regarding internal contextual grammars with suffix-closed selection languages.
\end{abstract}


\section{Introduction}

Contextual grammars, first proposed by Solomon Marcus  \cite{Marcus.1969} provide a formal framework 
for modeling the generation of natural languages. In this model, derivations proceed by adjoining 
pairs of `contexts'~--that is, ordered pairs of words \((u,v)\) -- to existing well-formed sentences. 
Specifically, an external application of a context \((u,v)\) to a word \(x\) yields the word \(u\,x\,v\), 
whereas an internal application produces every word of the form \(x_1\,u\,x_2\,v\,x_3\) for which \(x_1x_2x_3 = x\). 
To regulate the derivational process, each context is associated with a `selection' language: a context \((u,v)\) 
may only be applied around a word~\(x\) if \(x\) belongs to its designated selection language. By constraining 
selection languages to belong to a prescribed family~\(F\), one obtains contextual grammars with selection in \(F\).

The initial investigations into external contextual grammars with regular selection languages were conducted 
by J\"urgen Dassow \cite{Dassow.2005} and were subsequently extended -- both for external and internal variants 
by J\"urgen Dassow, Florin Manea, and Bianca Truthe \cite{Dassow_Manea_Truthe.2012,DasManTru12b}. These studies 
examined the impact of various subregular restrictions on selection languages. In the present work, we further 
refine this hierarchy by introducing families of `idefix-closed' subregular languages and explore the generative 
power of both external and internal contextual grammars whose selection languages lie in these newly defined 
families. With `idefix-closed', we mean prefix-, suffix-, or infix-closed.

Especially, in the present paper, we solve an open problem regarding internal contextual grammars with 
suffix-closed selection languages which was raised already several years ago in \cite{Truthe.2017}.

\section{Preliminaries}

Throughout the paper, we assume that the reader is familiar with the basic concepts of the theory of automata 
and formal languages. For details, we refer to \cite{Rozenberg_Salomaa.1997}. Here we only recall some notation, 
definitions, and previous results which we need for the present research.

An alphabet is a non-empty finite set of symbols. For an alphabet $V$, we denote by~$V^*$ and $V^+$ the set of all 
words and the set of all non-empty words over $V$, respectively. The empty word is denoted by~$\lambda$. 
For a word $w$ and a letter $a$, we denote the length of $w$ by $|w|$ and the number
of occurrences of the letter~$a$ in the word $w$ by $|w|_a$. For a set $A$, we denote its cardinality by $|A|$.

The family of the regular languages is denoted by $\REG$.
Any subfamily of this set is called a subregular language family.

For a language $L$ over an alphabet $V$, we set
\begin{align*}
\Inf(L) &=\set{y}{xyz\in L \text{ for some } x,z\in V^*},\\
\Pre(L) &=\set{x}{xy\in L \text{ for some } y\in V^*},\\
\Suf(L) &=\set{y}{xy\in L \text{ for some } x\in V^*}
\end{align*}
as the infix-, prefix-, and suffix-closure of $L$, respectively.
If the language $L$ is regular, then also $\Inf(L)$ and $\Suf(L)$ are regular.

\subsection{Some Subregular Language Families}

We consider the following restrictions for regular languages. In the following list of properties, we give already 
the abbreviation which denotes the family of all languages with the respective property. 
Let $L$ be a regular 
language over an alphabet $V$. With respect to the alphabet $V$, the language $L$ is said to be
\begin{itemize}
\item \emph{monoidal} ($\MON$) if and only if $L=V^*$,
\item \emph{nilpotent} ($\NIL$) if and only if it is finite or its complement $V^*\setminus L$ is finite,
\item \emph{combinational} ($\COMB$) if and only if it has the form
$L=V^*X$
for some subset~$X\subseteq V$,
\item \emph{definite} ($\DEF$) if and only if it can be represented in the form
$L=A\cup V^*B$
where~$A$ and~$B$ are finite subsets of $V^*$,
\item \emph{symmetric definite} ($\SYDEF$) if and only if $L = EV^*H$ for some regular languages $E$ and $H$,
\item \emph{infix-closed} ($\INF$) if
and only if, for any three words over $V$, say $x\in V^*$, $y\in V^*$ and $z \in V^*$, the relation $xyz\in L$ implies
the relation~$y\in L$,
\item \emph{prefix-closed} ($\PRE$) if
and only if, for any two words over $V$, say $x\in V^*$ and $y \in V^*$, the relation $xy\in L$ implies
the relation~$x\in L$,
\item \emph{suffix-closed} ($\SUF$) if
and only if, for any two words over $V$, say $x\in V^*$ and~$y\in V^*$, the relation $xy\in L$ implies
the relation~$y\in L$,
\item \emph{ordered} ($\ORD$) if and only if the language is accepted by some deterministic finite
automaton 
\[{\cal A}=(V,Z,z_0,F,\delta)\]
with an input alphabet $V$, a finite set $Z$ of states, a start state $z_0\in Z$, a set $F\subseteq Z$ of
accepting states and a transition mapping $\delta$ where $(Z,\preceq )$ is a totally ordered set and, for
any input symbol~$a\in V$, the relation $z\preceq z'$ implies~$\delta (z,a)\preceq \delta (z',a)$,
\item \emph{commutative} ($\COMM$) if and only if it contains with each word also all permutations of this
word,
\item \emph{circular} ($\CIRC$) if and only if it contains with each word also all circular shifts of this
word,
\item \emph{non-counting} ($\NC$) if and only if there is a natural
number $k\geq 1$ such that, for any three words~$x\in V^*$, $y\in V^*$, and $z\in V^*$, it 
holds~$xy^kz\in L$ if and only if~$xy^{k+1}z\in L$,
\item \emph{star-free} ($\SF$) if and only if $L$ can be described by a regular expression which is built by 
concatenation, union, and complementation,     
\item \emph{power-separating} ($\PS$) if and only if, there is a natural number $m\geq 1$ such that
for any word~$x\in V^*$, either
$J_x^m \cap L = \emptyset$
or
$J_x^m\subseteq L$
where
$J_x^m = \set{ x^n}{n\geq m}$,
\item \emph{union-free} ($\UF$) if and only if $L$ can be described by a regular expression which
is only built by concatenation and Kleene closure,
\item \emph{star} ($\STAR$) if and only if $L = H^*$ for some regular language $H \subseteq V^*$,
\item \emph{left-sided comet} ($\LCOM$) if and only if $L = EG^*$ for some regular language $E$ and a regular 
language $G \notin \{\emptyset, \{\lambda\}\}$,
\item \emph{right-sided comet} ($\RCOM$) if and only if $L = G^*H$ for some regular language $H$ and a regular 
language $G \notin \{\emptyset, \{\lambda\}\}$,
\item \emph{two-sided comet} ($\TCOM$) if and only if $L = EG^*H$ for two regular languages $E$ and $H$ and a 
regular language $G \notin \{\emptyset, \{\lambda\}\}$.
\end{itemize}

We group the language families $\SUF$, $\PRE$ and $\INF$ under the term idefix-closed families.
We remark that monoidal, nilpotent, combinational, (symmetric) definite, ordered, star-free, 
union-free, star, and (left-, right-, or two-sided) comet languages are regular, whereas non-regular languages 
of the other types mentioned above exist.
Here, we consider among the suffix-closed, commutative, circular, non-counting, 
and power-separating languages only those which are also regular.
By $\FIN$, we denote the family of languages with finitely many words.
In \cite{McNaughton_Papert.1971}, it was shown that the families of the regular non-counting languages and 
the star-free languages are equivalent~($\NC=\SF$).

Some properties of the languages of the classes mentioned above can be found in
\cite{Shyr.1991} (monoids),
\cite{Gecseg_Peak.1972} (nilpotent languages),
\cite{Havel.1969} (combinational and commutative languages),
\cite{Perles_Rabin_Shamir.1963} (definite languages),
\cite{Paz_Peleg.1965} (symmetric definite languages),
\cite{Brzozowski_Jiraskova_Zou.2014} (prefix-closed languages),
\cite{Gill_Kou.1974} and \cite{Brzozowski_Jiraskova_Zou.2014} (suffix-closed languages),
\cite{Shyr_Thierrin.1974.ord} (ordered languages),
\cite{Kudlek.2004} (circular languages),
\cite{McNaughton_Papert.1971} (non-counting and 
star free
languages),
\cite{Shyr_Thierrin.1974.ps} (power-separating languages),
\cite{Brzozowski.1962} (union-free languages),
\cite{Brzozowski.1967} (star languages),
\cite{Brzozowski_Cohen.1969} (comet languages).

\subsection{Contextual Grammars}

Let $\cF$ be a family of languages. A contextual grammar with selection in $\cF$ is a triple~$G=(V,\cS,A)$
with the following components:
\begin{btlists}
\item $V$ is an alphabet. 
\item $\cS$ is a finite set of selection pairs $(S,C)$ where $S$ is called selection language and $C$ is a set of 
so-called contexts. Any selection language $S$ is a language in the family $\cF$ with respect to an alphabet~$U \subseteq V$. 
Any set $C$ of contexts is a finite set~\hbox{$C\subset V^*\times V^*$} where, for each context~$(u,v)\in C$, 
at least one side is not empty: $uv\not=\lambda$.
\item $A$ is a finite subset of $V^*$ (its elements are called axioms).
\end{btlists}
We write a selection pair $(S,C)$ also as $S\to C$. In the case that $C$ is a singleton set~$C=\{(u,v)\}$, we
also write $S\to(u,v)$.

For a contextual grammar 
$G=(V,\sets{S_1\to C_1, S_2\to C_2,\dots, S_n\to C_n},A)$,
we set
\[\ellA(G) = \max \Set{ |w| }{ w\in A },\
\ellC(G) = \max \Set{ |uv| }{(u,v)\in C_i, 1\leq i\leq n},\ 
\ell(G)   = \ellA(G)+\ellC(G)+1.\]
We now define the derivation modes for contextual grammars with selection.

Let $G=(V,\cS,A)$ be a contextual grammar with selection.
A direct external derivation step in $G$ is defined as follows: a word~$x$ derives a word $y$ 
(written as~$x\Lra_\mathrm{ex} y$) if and only if there is a pair~$(S,C)\in\cS$ such that~$x\in S$ and $y=uxv$ 
for some pair $(u,v)\in C$.
Intuitively, one can only wrap a context $(u,v)\in C$ around a word $x$ if $x$ belongs to the corresponding
selection language $S$.

A direct internal derivation step in $G$ is defined as follows: a word~$x$ derives a word~$y$ 
(written as~$x\Lra_\mathrm{in} y$) if and only if there are words $x_1$,~$x_2$,~$x_3$
with~$x_1x_2x_3=x$ and there is a selection pair~$(S,C)\in\cS$ such that $x_2\in S$ and~$y=x_1ux_2vx_3$ 
for some pair $(u,v)\in C$.
Intuitively, we can only wrap a context~$(u,v)\in C$ around a subword~$x_2$ of~$x$ if $x_2$ belongs to 
the corresponding selection language~$S$.
%

By $\Lra^*_\mu$ we denote the reflexive and transitive closure of the relation~$\Lra_\mu$ 
for~$\mu\in\sets{\mathrm{ex},\mathrm{in}}$.
The language generated by $G$ is defined as
\[
L_\mu(G)=\set{ z }{ x\Lra^*_\mu z \mbox{ for some } x\in A }.
\]
We omit the index $\mu$ if the derivation mode is clear from the context.

By~$\cEC(\cF)$, we denote the family of all languages generated externally by contextual grammars 
with selection in $\cF$. When a contextual grammar works in the external mode, we call it an external 
contextual grammar. 
By~$\cIC(\cF)$, we denote the family of all languages generated internally by contextual grammars 
with selection in $\cF$. When a contextual grammar works in the internal mode, we call it an internal 
contextual grammar. 

\section{Results on families of idefix-closed languages
}

In this section, we investigate inclusion relations between various subregular languages classes.
Figure~\ref{fig:lang_erg_1} shows the results. 

\begin{figure}[htb]
  \centering
  \scalebox{.9}{\begin{tikzpicture}[node distance=15mm and 16mm, on grid]
    \node (MON) {$\MON$};
    \node (d0)[above=of MON] {};
    \node (FIN)[right=of d0] {$\FIN$};
    \node (NIL)[above=of d0] {$\NIL$};
    \node (COMB)[left=of NIL] {$\COMB$};
    \node (DEF)[above=of NIL] {$\DEF$};
    \node (d1)[left=of DEF] {};
    \node (SYDEF)[left=of d1] {$\SYDEF$};
    \node (d2)[right=of NIL] {};
    \node (ORD)[above=of DEF] {$\ORD$};
    \node (INF)[right=of d2] {$\INF$};
    \node (PRE)[right=of ORD] {$\PRE$};
    \node (SUF)[right=of PRE] {$\SUF$};
    \node (NC)[above=of ORD] {$\NC\stackrel{\text{\cite{McNaughton_Papert.1971}}}{=}\SF$};
    \node (PS)[above=of NC] {$\PS$};
    \node (RCOM)[above=of SYDEF] {$\RCOM$};
    \node (LCOM)[left=of RCOM] {$\LCOM$};
    \node (TCOM)[above=of RCOM] {$\TCOM$};
    \node (COMM)[right=of SUF] {$\COMM$};
    \node (CIRC)[above=of COMM] {$\CIRC$};
    \node (UF)[left=of LCOM] {$\UF$};
    \node (STAR)[below=of UF] {$\STAR$};
    \node (REG) [above of = PS] {$\REG$};
  
    \draw[hier, bend left] (MON) to node[edgeLabel] {\small\cite{Koedding.Truthe.2024}} (STAR);
    \draw[hier, bend left] (MON) to node[edgeLabel] {\small\cite{Koedding.Truthe.2024}} (SYDEF);
    \draw[hier] (RCOM) to node[pos=.45,edgeLabel]{\small\cite{Bordihn_Holzer_Kutrib.2009}}(TCOM);
    \draw[hier] (LCOM) to node[edgeLabel]{\small\cite{Koedding.Truthe.2024}}(TCOM);
    \draw[hier, bend left=20] (TCOM) to node[edgeLabel]{\small\cite{Bordihn_Holzer_Kutrib.2009}}(REG);
    \draw[hier] (STAR) to node[pos=.45,edgeLabel]{\small\cite{Koedding.Truthe.2024}} (UF);
    \draw[hier, bend left] (UF) to node[edgeLabel]{\small\cite{Holzer_Truthe.2015}}(REG);
    \draw[hier, bend right=40] (MON) to node[pos=.7,edgeLabel]{\small\cite{Truthe.2018}} (COMM);
    \draw[hier] (COMM) to node[pos=.45,edgeLabel]{\small\cite{Holzer_Truthe.2015}}(CIRC);
    \draw[hier, bend right] (CIRC) to node[edgeLabel]{\small{\cite{Holzer_Truthe.2015}}}(REG);
    \draw[hier] (MON) to node[pos=.45,edgeLabel]{\small\cite{Truthe.2018}} (NIL);
    \draw[hier] (NIL) to node[pos=.45,edgeLabel]{\small\cite{Wiedemann.1978}}(DEF);
    \draw[hier] (DEF) to node[pos=.45,edgeLabel]{\small\cite{Holzer_Truthe.2015}}(ORD);
    \draw[hier] (ORD) to node[pos=.45,edgeLabel]{\small\cite{Shyr_Thierrin.1974.ord}}(NC);
    \draw[hier] (NC) to node[pos=.4,edgeLabel]{\small\cite{Shyr_Thierrin.1974.ps}}(PS);
    \draw[hier] (PS) to node[pos=.45,edgeLabel]{\small\cite{Holzer_Truthe.2015}}(REG);
    \draw[hier, bend right] (MON) to node[pos=.7,edgeLabel]{\small \ref{srl:subsets}}(INF);
    \draw[hier, bend right] (SUF) to node[edgeLabel]{\small\cite{Holzer_Truthe.2015}}(PS);
    \draw[hier] (FIN) to node[pos=.45,edgeLabel]{\small\cite{Wiedemann.1978}}(NIL);
    \draw[hier] (COMB) to node[pos=.45,edgeLabel]{\small\cite{Havel.1969}}(DEF);
    \draw[hier] (COMB) to node[pos=.45,edgeLabel]{\small\cite{Olejar_Szabari.2023}}(SYDEF);
    \draw[hier] (SYDEF) to node[pos=.45,edgeLabel]{\small\cite{Koedding.Truthe.2024}}(LCOM);
    \draw[hier] (SYDEF) to node[pos=.45,edgeLabel]{\small\cite{Olejar_Szabari.2023}}(RCOM);
    \draw[hier] (SYDEF) to node[edgeLabel]{\small\cite{Olejar_Szabari.2023}}(PS);

    \draw[hier, bend left] (INF) to node[edgeLabel]{\small \ref{srl:subsets}}(PRE);
    
    \draw[hier] (INF) to node[edgeLabel]{\small \ref{srl:subsets}}(SUF);
    
    \draw[hier, bend right] (PRE) to node[edgeLabel]{\small \ref{srl:subsets_pre_ps}}(PS);
  \end{tikzpicture}}
  \caption{Resulting hierarchy of subregular language families.}
\label{fig:lang_erg_1}
  \end{figure}
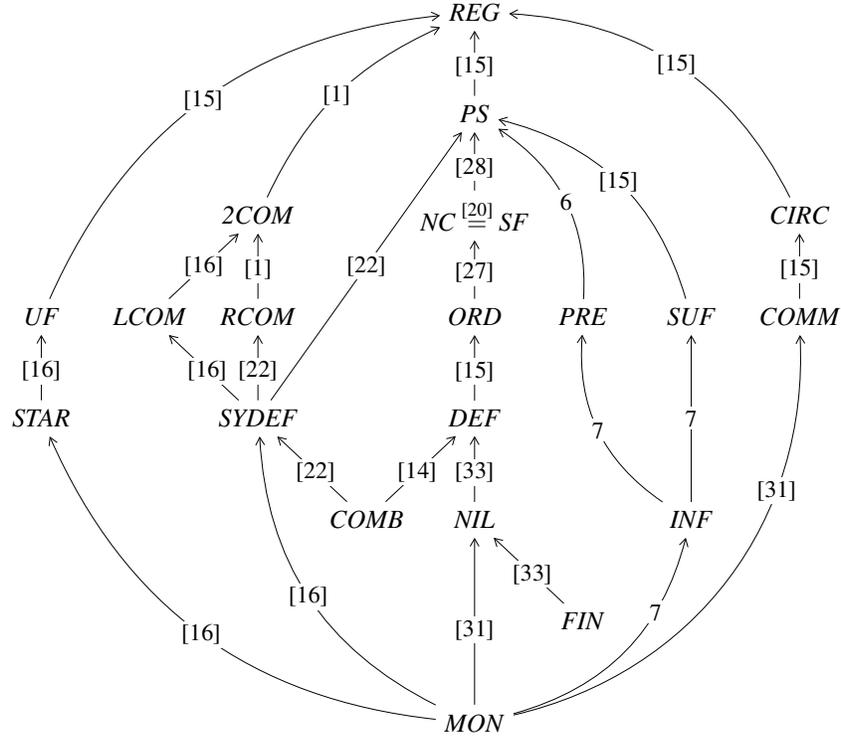

An arrow from a node $X$ to a node~$Y$ stands for the proper inclusion $X \subset Y$. 
  If two families are not connected by a directed path, they are incomparable. 
  An edge label refers to the paper where the proper inclusion has been shown (in some cases, it might be that it is not the 
  first paper where the respective inclusion has been mentioned, since it is so obvious that it was not emphasized in a 
  publication) or the lemma of this paper where the proper inclusion will be shown.

  In the literature, it is often said that two languages are equivalent if they are equal or differ
  at most in the empty word. Similarly, two families can be regarded to be equivalent if they differ 
  only in the languages $\emptyset$ or $\{\lambda\}$. Therefore, the set $\STAR$ of all star languages is 
  sometimes regarded as a proper subset of the set $\COM$ of all (left-, right-, or two-sided) comet languages 
  although $\{\lambda\}$ belongs to the family $\STAR$ but not to $\LCOM$, $\RCOM$ 
  or $\TCOM$. We regard $\STAR$ and $\STAR\setminus\{\{\lambda\}\}$ as different.

%

  We now present some languages which will serve later as witness languages for proper inclusions or
  incomparabilities.\pagebreak
  
\begin{lemma}\label{srl:pre_o_suf}
    Let $L = \{ab, a, \lambda\}$. Then, it holds $L \in (\FIN \cap \PRE) \setminus \SUF$.
\end{lemma}
\begin{proof}
    For the word $ab$, the prefixes $a$ and $\lambda$ are in the language. For the word $a$, 
    the prefix $\lambda$ is in the language. Furthermore, $L$ is finite. Therefore, $L \in \FIN \cap \PRE$.
    For the word $ab$, the suffix $b$ is not in the language. Therefore, $L \notin \SUF$ holds.
\end{proof}

\begin{lemma}\label{srl:fin_suf_o_pre}
    Let $L = \{ab, b, \lambda\}$ . Then, it holds $L \in (\FIN \cap \SUF) \setminus \PRE$.
\end{lemma}
\begin{proof}
    For the word $ab$, the suffixes $b$ and $\lambda$ are in the language. For the word $b$, the suffix $\lambda$ 
    is in the language. Furthermore, $L$ is finite. Therefore, $L \in \SUF$.
    For the word $ab$, the prefix $a$ is not in the language. Therefore, $L \notin \PRE$ holds.
\end{proof}

\begin{lemma}\label{srl:inf_o_uf_circ_tcom}
    Let $L = \{ab, a, b, \lambda\}$ . Then, it holds $L \in \INF \setminus (\UF \cup \CIRC \cup \TCOM)$.
\end{lemma}
\begin{proof}
    For the word $ab$, all infixes $ab$, $a$, $b$ and~$\lambda$ are in the language. For the word $b$, all infixes $b$ and~$\lambda$ are also in the language. For the word $a$, all infixes $a$ and $\lambda$ are in the language, and for the word~$\lambda$, the infix $\lambda$ is in the language. It therefore holds that for every word, all infixes are in the language $L$, so~$L \in \INF$.
    
    For the word $ab$, the circular permutation $ba$ is not in the language. Therefore, $L \notin \CIRC$.
    According to 
    \cite{Nagy.2019}, a language has an infinite number of words or at most one word if it is union-free. However, this language has four words and therefore $L \notin \UF$ holds.
    Every non-empty language from the family $\TCOM$ is infinite by 
    \cite{Koedding.Truthe.2024}. Since $L$ is non-empty but finite, we have~$L \not\in \TCOM$.
\end{proof}

\begin{lemma}\label{srl:inf_o_nc}
    Let $L = \Inf(\set{ab^{2n}a}{n \geq 1})$. Then, it holds $L \in \INF \setminus \NC$.
\end{lemma}
\begin{proof}
Since $L$ is the infix-closure of a language, we see $L\in\INF$.

Assuming that $L$ is non-counting, it follows from the definition that for all words $x,y,z \in \{a,b\}^*$ and 
for a number $k \geq 1$, the equivalence 
$xy^kz \in L \Longleftrightarrow xy^{k+1}z \in L$
applies. We now set $x = z = a$ and  $y = b$. 

If $k$ is even, $ab^ka \in L$ but~$ab^{k+1}a \in L$, which is a contradiction. If $k$ is odd, $ab^{k+1}a \in L$ 
but~$ab^ka \in L$, which is also a contradiction. It follows that $L \notin \NC$ holds.
\end{proof}

\begin{lemma}\label{srl:comb_o_pre}
    Let 
    $L = \{a,b\}^*\{b\}$. Then, it holds $L \in \COMB \setminus \PRE$.
\end{lemma}
\begin{proof}
    With $V = \{a,b\}$ and 
    $A = \{b\}$, the language $L$ has the structure $L = V^*A$. Therefore, $L \in \COMB$.
    The word $aab$ is in $L$ but the prefix $aa$ is not. Hence, $L \notin \PRE$ holds.
\end{proof}

%

We now prove some inclusion relations.

\begin{lemma}\label{srl:subsets_pre_ps}
    The proper inclusion $\PRE \subset \PS$ holds.
\end{lemma}
\begin{proof}
The inclusion $\PRE \subseteq \PS$ can be shown similarly to the inclusion $\SUF \subset \PS$ which was proved 
in~\cite{Holzer_Truthe.2015}.
The language $L = \{a,b\}^*\{b\}$ from Lemma~\ref{srl:comb_o_pre} is a witness language since it is in $\COMB \setminus \PRE$ and therefore in $\PS \setminus \PRE$.
\end{proof}

\begin{lemma}\label{srl:subsets}
    The proper inclusions $\MON \subset \INF \subset \PRE$ and $\INF \subset \SUF$ hold.
\end{lemma}
\begin{proof*}
The inclusions
$\MON \subseteq \INF \subseteq \PRE$ and $\INF \subseteq \SUF$
follow from the definition. For their properness, we have the following witness languages:
    \begin{enumerate}
        \item $\MON \subset \INF$: The language $L = \Inf\left(\set{ab^{2n}a}{n\geq 1}\right)$ from Lemma~\ref{srl:inf_o_nc} is a witness language since it is in $\INF \setminus \NC$ and therefore in $\INF \setminus \MON$ (because $\MON\subset\NC$).
        \item $\INF \subset \PRE$: The language $L = \{ab,a,\lambda\}$ from Lemma~\ref{srl:pre_o_suf} is a witness language since it belongs to the set $\PRE \setminus \SUF$ and therefore, 
        it also belongs to the set $\PRE \setminus \INF$.
        \item $\INF \subset \SUF$: The language $L = \{ab,b,\lambda\}$ from Lemma~\ref{srl:fin_suf_o_pre} is a witness language since it belongs to the set $\SUF \setminus \PRE$ and therefore, 
        it also belongs to the set $\SUF \setminus \INF$.\hfill$\Box$
    \end{enumerate}
\end{proof*}

We now prove the incomparability relations mentioned in Figure~\ref{fig:lang_erg_1} which have not been proved earlier. 
These are the relations regarding the families $\PRE$ and $\INF$. 
In most cases, we show the incomparability of whole `strands' in the hierarchy. For a strand consisting of language
families $F_1, F_2, \dots, F_n$ where $F_1$ is a subset of every family $F_i$ with $1\leq i\leq n$ and every such family
is a subset of the family $F_n$ and a strand consisting of families $F'_1, F'_2, \dots, F'_m$ where $F'_1\subseteq F'_j$
and $F'_j\subseteq F'_m$ for $1\leq j\leq m$, it suffices to show that there are a language $L$ in $F_1 \setminus F'_m$ 
and a language $L'$ in $F'_1 \setminus F_n$ in order to show that every family $F_i$ is incomparable to every family $F'_j$
with $1\leq i\leq n$ and $1\leq j\leq m$ (because $L\in F_i \setminus F'_j$ and $L'\in F'_j \setminus F_i$).
So, we give only two witness languages $L$ and $L'$ for every pair of strands.

\begin{lemma}\label{lemma:pre_unvergleichbarzu_suf}
  The language families $\SUF$ and $\PRE$ are incomparable to each other.
\end{lemma}
\begin{proof}
Witness languages are given in Lemmas~\ref{srl:pre_o_suf} and~\ref{srl:fin_suf_o_pre}.
\end{proof}

\begin{lemma}\label{lemma:pre_inf_unvergleichbarzu_comb_nc}
    Let $\cF = \{\COMB, \DEF, \SYDEF, \ORD, \NC\}$. 
  Every family in $\cF$ is incomparable to the families ${\PRE}$ and $\INF$.
\end{lemma}
\begin{proof}
    As witness languages, we have 
    \[L_1 = \Inf\left(\set{ab^{2n}a}{n \geq 1}\right) \in \INF \setminus \NC \qmand
    L_2 = \{a,b\}^*\{b\} \in \COMB \setminus \PRE\]
    from Lemma~\ref{srl:inf_o_nc} and 
    Lemma~\ref{srl:comb_o_pre}, respectively.
\end{proof}

\begin{lemma}\label{lemma:pre_inf_unvergleichbarzu_fin_nil}
    Let $\cF = \{\FIN, \NIL\}$. 
  Every family in $\cF$ is incomparable to the families ${\PRE}$ and $\INF$.
\end{lemma}
\begin{proof}
    As witness languages, we have 
    \[L_1 = \{ab,b,\lambda\} \in \FIN \setminus \PRE \qmand
    L_2  = \Inf\left(\set{ab^{2n}a}{n \geq 1}\right) \in \INF \setminus \NC\]
    from Lemma~\ref{srl:fin_suf_o_pre} and Lemma~\ref{srl:inf_o_nc}, respectively.
\end{proof}

\begin{lemma}\label{lemma:pre_inf_unvergleichbarzu_sydef_2com}
    Let $\cF = \{\SYDEF, \RCOM, \LCOM, \TCOM\}$. 
  Every family in $\cF$ is incomparable to the families ${\PRE}$ and $\INF$.
\end{lemma}
\begin{proof}
    As witness languages, we have 
    \[L_1 = \{ab,a,b,\lambda\} \in \INF \setminus \TCOM \qmand 
    L_2 = \{a,b\}^*\{b\} \in \COMB \setminus \PRE\]
    from Lemma~\ref{srl:inf_o_uf_circ_tcom} and Lemma~\ref{srl:comb_o_pre}, respectively.
\end{proof}

\begin{lemma}\label{lemma:pre_inf_unvergleichbarzu_star_uf}
    Let $\cF = \{\STAR, \UF\}$. 
  Every family in $\cF$ is incomparable to the families ${\PRE}$ and $\INF$.
\end{lemma}
\begin{proof}
    As witness languages, we have 
    \[L_1 = \{ab,a,b,\lambda\} \in \INF \setminus \UF \qmand
    L_2 = \{aa\}^* \in \STAR \setminus \PS\]
    from Lemma~\ref{srl:inf_o_uf_circ_tcom} and \cite[Lemma~11]{Koedding.Truthe.2024}, respectively.
\end{proof}

\begin{lemma}\label{lemma:pre_inf_unvergleichbarzu_comm_circ}
    Let $\cF = \{\COMM, \CIRC\}$. 
  Every family in $\cF$ is incomparable to the families ${\PRE}$ and $\INF$.
\end{lemma}
\begin{proof}
    As witness languages, we have 
    \[L_1 = \{ab,a,b,\lambda\} \in \INF \setminus \CIRC \qmand
    L_2 = \{aa\}^* \in \COMM \setminus \PS
    \]
    from Lemma~\ref{srl:inf_o_uf_circ_tcom} and \cite[Lemma~4.9]{Holzer_Truthe.2015}.
\end{proof}

From all these relations, the hierarchy presented in Figure~\ref{fig:lang_erg_1} follows.
An edge label refers to the paper or lemma in the present paper where the proper inclusion is shown.
The incomparability results are proved in Lemmas~\ref{lemma:pre_unvergleichbarzu_suf} 
through~\ref{lemma:pre_inf_unvergleichbarzu_comm_circ}.

\begin{theorem}[Resulting hierarchy for subregular families]\label{theorem:neue_hierarchie}
The inclusion relations presented in Figure \ref{fig:lang_erg_1} hold. An arrow from an entry $X$ to
an entry~$Y$ depicts the proper inclusion $X \subset Y$; if two families are not connected by a directed
path, they are incomparable.
\end{theorem}

\section{Results on subregular control in external contextual grammars}

In this section, we include the families of languages generated by external contextual grammars with 
selection languages from the subregular families under investigation into the existing hierarchy 
with respect to external contextual grammars.


\begin{lemma}[Monotonicity $\cEC$]\label{lemma:ec_monoton}
  For any two language classes $X$ and $Y$ with~$X\subseteq Y$,
  we have the inclusion~$\cEC(X)\subseteq\cEC(Y)$.
\end{lemma}

Figure~\ref{fig:ec_erg} shows 
the inclusion relations between language families 
which are generated by external
contextual grammars where the selection languages belong to subregular classes investigated before. The 
hierarchy contains results which were already known (marked by a reference to the literature) and results which 
are new.

  \begin{figure}[htb]
  \centering
  \scalebox{.77}{
  \begin{tikzpicture}[node distance=15mm and 25mm,on grid=true
  ]
    \node (MON) {$\ec{\MON}$};
    \node (FIN) [below=of MON] {$\ec{\FIN}$};
    \node (COMB)[above=of MON] {$\ec{\COMB}$};
    \node (NIL) [right=of COMB] (NIL) {$\ec{\NIL}$};
    \node (DEF) [above=of COMB] {$\ec{\DEF}$};
    \node (ORD) [above=of DEF] {$\ec{\ORD}$};
    \node (SYDEF) [right=of ORD] {$\ec{\SYDEF}$};
    \node (INF) at (-3.7, 3) {$\ec{\INF}$};
    \node (SUF) [left=of ORD] {$\ec{\SUF}$};
    \node (PRE) [left=of SUF] {$\ec{\PRE}$};
    
    \node (COMM) [right=of SYDEF] {$\ec{\COMM}$};
    \node (STAR) [below right=of COMM] {$\ec{\STAR}$};
    \node (NC) [above=of ORD] {$\ec{\NC}$};
    \node (PS) [above=of NC] {$\ec{\PS}$};
    \node (CIRC) [above=of COMM] {$\ec{\CIRC}$};
    \node (REG) [above=of PS] {$\ec{\REG} \stackrel{\text{\cite{Dassow_Manea_Truthe.2012}}}{=} \ec{\UF} \stackrel{\text{\cite{Koedding.Truthe.2024}}}{=} \cEC(\LCOM) \stackrel{\text{\cite{Koedding.Truthe.2024}}}{=} \cEC(\RCOM) \stackrel{\text{\cite{Koedding.Truthe.2024}}}{=} \cEC(\TCOM)$};

    \draw[hier] (FIN) to node[pos=.45, edgeLabel]{\small\cite{Dassow.2005}} (MON);
    \draw[hier, bend right] (MON) to node[pos=.45, edgeLabel]{\small\cite{Dassow.2005}} (NIL);
    \draw[hier] (MON) to node[pos=.45, edgeLabel]{\small\cite{Dassow.2015}} (COMB);
    \draw[hier] (COMB) to node[pos=.45, edgeLabel]{\small\cite{Truthe.2021}} (DEF);
    \draw[hier] (DEF) to node[pos=.45, edgeLabel]{\small\cite{Truthe.2014}} (ORD);
    \draw[hier, bend right=15] (DEF) to node[pos=.45, edgeLabel]{\small\cite{Koedding.Truthe.2024}} (SYDEF);
    \draw[hier] (ORD) to node[pos=.45, edgeLabel]{\small\cite{Dassow_Truthe.2023}} (NC);
    \draw[hier] (NC) to node[pos=.45, edgeLabel]{\small\cite{Truthe.2021}} (PS);
    \draw[hier, bend left=30] (MON) to 
    (INF);
    \draw[hier, bend right=30] (INF) to 
    (SUF);
    \draw[hier, bend left=30] (INF) to 
    (PRE);
    \draw[hier, bend left=40] (SUF) to node[pos=.45, edgeLabel]{\small\cite{Truthe.2021}} (PS);
    \draw[hier, bend left=30] (PRE) to 
    (PS);
    \draw[hier, bend right=35] (MON) to node[pos=.45, edgeLabel]{\small\cite{Koedding.Truthe.2024}} (STAR);
    \draw[hier, bend right=35] (STAR) to node[pos=.45, edgeLabel]{\small\cite{Koedding.Truthe.2024}} (REG);
    \draw[hier, bend right] (SYDEF) to node[pos=.45, edgeLabel]{\small\cite{Koedding.Truthe.2024}} (PS);
    \draw[hier, bend right] (NIL) to node[pos=.45, edgeLabel]{\small\cite{Dassow.2005}} (COMM);
    \draw[hier, bend right] (NIL) to node[pos=.45, edgeLabel]{\small \cite{Dassow.2005}} (DEF);
    \draw[hier] (COMM) to node[pos=.45, edgeLabel]{\small\cite{Dassow_Manea_Truthe.2012}} (CIRC);
    \draw[hier, bend right=20] (CIRC) to node[pos=.45, edgeLabel]{\small\cite{Dassow_Manea_Truthe.2012}} (REG);
    \draw[hier] (PS) to node[pos=.45, edgeLabel]{\small\cite{Truthe.2021}} (REG);
  \end{tikzpicture}
  }
  \caption{Resulting hierarchy of language families by external contextual grammars with special selection languages.}
\label{fig:ec_erg}
  \end{figure}

We now present some languages which 
serve 
as witness languages for proper inclusions or
incomparabilities.

\begin{lemma}\label{ec:pre_o_suf}
    The language $L = \set{a^nb^n}{n \geq 1} \cup \set{b^n}{n \geq 1}$ is in $\ec{\PRE} \setminus \ec{\SUF}$.
\end{lemma}
\begin{proof}
    The contextual grammar $G = (\{a,b\}, \{S_1 \to C_1, S_2 \to C_2\}, \{ab,b\})$ with
\[S_1=\Pre(\set{a^nb^m}{n,m\geq 1}),\
C_1=\sets{(a,b)},\ S_2=\sets{b}^*,\ C_2=\sets{(\lambda, b)}
\]
    generates the language $L$ as can be seen as follows.

    The first rule can be applied to the axiom $ab$ and then to every other word $a^kb^k$ for $k\geq 1$. This allows us to generate the language $\set{a^nb^n}{n\geq 1}$. From the axiom $ab$, no other word can be generated.
    The second rule can be applied to the axiom $b$ and then to every other word $b^k$ for $k\geq 1$. This allows us to generate the language $\set{b^n}{n\geq 1}$. From the axiom $b$, no other word can be generated.
    Together, we obtain that $L(G) = L$.
    Furthermore, all selection languages are prefix-closed. Hence, $L\in\ec{\PRE}$.
    From \cite[Lemma 3.3]{Dassow.2005}, we know that the language $L$ is not in $\ec{\SUF}$.
\end{proof}

\begin{lemma}\label{ec:suf_o_pre}
    The language $L = \set{a^nb^n}{n \geq 1} \cup \set{a^n}{n \geq 1}$ is in $\ec{\SUF} \setminus \ec{\PRE}$.
\end{lemma}
\begin{proof}
The proof is similar to the one for Lemma~\ref{ec:pre_o_suf} due to the symmetry.
\end{proof}

\begin{lemma}\label{ec:inf_o_star}
The language $L=\sets{a,b}^*\set{a^nb^m}{n\geq 1, m\geq 1} \cup \set{ca^nb^mc}{n\geq 1, m\geq 1}$
belongs to the family $\ec{\INF} \setminus \ec{\STAR}$.
\end{lemma}
\begin{proof}
    The contextual grammar $G = (\{a,b\}, \{S_1 \to C_1, S_2 \to C_2\}, \{ab\})$ with
\[S_1=\Inf(\set{a^nb^m}{n\geq 1, m\geq 1}),\ C_1=\{(c,c)\},\ S_2=\Inf(\sets{a,b}^*),\ C_2=\{(a, \lambda),(b, \lambda),(\lambda, b)\}\]
    generates the language $L$ where all selection languages are infix-closed.
With star languages as selection languages, the structure of a word cannot be checked before adjoining the letter $c$.
\end{proof}

\begin{lemma}\label{ec:inf_o_nc}
    Let
    $L= \set{a^mbc^{2n}ba^m}{n\geq 1, m \geq 0}\cup\set{c^n}{n \geq 2}\cup\set{bc^nb}{n \geq 2}\cup\set{ac^na}{n \geq 2}.$
    Then, $L \in \ec{\INF} \setminus \ec{\NC}$.
\end{lemma}
\begin{proof}
    The contextual grammar $G = (\{a,b,c\}, \{S_1 \to C_1, S_2 \to C_2\}, \{cc\})$ with 
    \[S_1=\{c\}^*,\ C_1=\{(\lambda, c),(b, b)\},\ S_2=\Inf(\set{a^mbc^{2n}ba^m}{n\geq 1, m \geq 0}),\ C_2=\sets{(a, a)}\]
    generates the language $L$ and all selection languages are infix-closed.
    With non-counting selection languages, one could also wrap letters $a$ around a word $bc^nb\in L$ for an odd number $n$ which is a contradiction.
\end{proof}

\begin{lemma}\label{ec:inf_o_sydef}
    The language $L = \{a,b\}^* \cup \{c\}\{\lambda, b\}\{ab\}^*\{\lambda, a\}\{c\}$ is in $\ec{\INF} \setminus \ec{\SYDEF}$.
\end{lemma}
\begin{proof}
    The language $L$ is generated by the contextual grammar 
  $G=(\sets{a,b,c},\sets{S_1\to C_1, S_2\to C_2},\sets{\lambda})$
  with 
  \[S_1 = \sets{a,b}^*,\ C_1 = \sets{(\lambda,a),(\lambda,b)},\
  S_2 = \Inf(\sets{ab}^*),\ C_2 = \sets{(c,c)}\]
%
  All selection languages are infix-closed; therefore,  we have $L\in\cEC(\INF)$.
  With symmetric definite selection languages, the alternation between $a$ and $b$ cannot be checked.
\end{proof}

\begin{lemma}\label{ec:inf_o_circ}
    The language $L = \set{a^nb^n}{n \geq 1} \cup \set{b^na^n}{n \geq 1}$ is in $\ec{\INF} \setminus \ec{\CIRC}$.
\end{lemma}
\begin{proof}
    The contextual grammar 
$G = (\{a,b\}, \{S_1 \to C_1, S_2 \to C_2\}, \{ab,ba\})$
with 
\[S_1=\Inf(\{a\}^*\{b\}^*),\ C_1=\sets{(a, b)},\ S_2=\Inf(\{b\}^*\{a\}^*),\ C_2=\sets{(b, a)}\]
generates $L$ where all selection languages are infix-closed.
With circular selection languages, a word of the language $\sets{a}^+\sets{b}^+\sets{a}^+\sets{b}^+$ could be generated.
\end{proof}

\begin{lemma}\label{ec:comb_o_pre}
    The language $L = \set{b^na}{n \geq 0}\cup \sets{\lambda}$ is in $\ec{\COMB} \setminus \ec{\PRE}$.
\end{lemma}
\begin{proof}
    The contextual grammar $G = (\{a,b\}, \{\{a,b\}^*\{a\} \to (b, \lambda)\},\sets{\lambda,a})$ generates the language $L$ and all the selection languages are combinational. 
With prefix-closed selection languages, also words without $a$ could be generated.
\end{proof}

\begin{lemma}\label{ec:nil_o_pre}
    The language $L = \set{bbba^n}{n\geq 1} \cup \sets{bb}$ is in $\ec{\NIL} \setminus \ec{\PRE}$.
\end{lemma}
\begin{proof}
    The contextual grammar $G = (\{a,b\}, \{\{a,b\}^4\{a,b\}^* \to (\lambda, a)\}, \{bbba, bb\})$ generates the language $L$, where all selection languages in~$\mathcal{S}$ are nilpotent.
With prefix-closed selection languages, also a word $bba^m$ could be generated.
\end{proof}

With the languages from the previous lemmas, the inclusion relations and incomparabilities depicted in Figure~\ref{fig:ec_erg} can be shown.
\begin{theorem}[Resulting hierarchy for $\cEC$]\label{theorem:neue_hierarchie_EC}
The inclusion relations presented in Figure \ref{fig:ec_erg} hold. An arrow from an entry $X$ to
an entry~$Y$ depicts the proper inclusion $X \subset Y$; if two families are not connected by a directed
path, they are incomparable.
\end{theorem}

\section{Results on subregular control in internal contextual grammars}

In this section, we include the families of languages generated by internal contextual grammars with 
selection languages from the subregular families under investigation into the existing hierarchy
with respect to internal contextual grammars.


\begin{lemma}[Monotonicity $\cIC$]\label{lemma:ic_monoton}
  For any two language classes $X$ and $Y$ with~$X\subseteq Y$,
  we have the inclusion~$\cIC(X)\subseteq\cIC(Y)$.
\end{lemma}

Figure~\ref{fig:ic_erg} shows a hierarchy of some language families which are generated by internal contextual grammars where the selection languages belong to subregular classes investigated before. The 
hierarchy contains results which were already known (marked by a reference to the literature) and 
results which 
are new.

  \begin{figure}[htb]
  \centering
  \scalebox{.9}{
  \begin{tikzpicture}[node distance=15mm and 25mm,on grid=true
  ]
    \node (MON) {$\ic{\MON}$};
    \node (X) [left=of MON] {};
    \node (FIN) [left=of X] {$\ic{\FIN}$};
    \node (NIL) [above=of FIN] {$\ic{\NIL}$};
    \node (COMB) [right=of NIL] {$\ic{\COMB}$};
    \node (DEF) [above=of NIL] {$\ic{\DEF}$};
    \node (ORD) [above=of DEF] {$\ic{\ORD}$};
    \node (NC) [above=of ORD] {$\ic{\NC}$};
    \node (PS) [above=of NC] {$\ic{\PS}$};
    \node (REG) at (0,9) {$\ic{\REG} \stackrel{\text{\cite{Dassow_Manea_Truthe.2012}}}{=} \ic{\UF}
\stackrel{\text{\cite{Koedding.Truthe.2024}}}{=} \ic{\LCOM} \stackrel{\text{\cite{Koedding.Truthe.2024}}}{=} \ic{\RCOM}
\stackrel{\text{\cite{Koedding.Truthe.2024}}}{=} \ic{\TCOM}$};

    \node (SYDEF) [right=of ORD] {$\ic{\SYDEF}$};

    \node (INF) [right=of SYDEF] {$\ic{\INF}$};

    \node (SUF) [above right=of INF] {$\ic{\SUF}$};
    
    \node (PRE) [left=of SUF] {$\ic{\PRE}$};

    \node (STAR) [right=of SUF] {$\ic{\STAR}$};

    \node (COMM) [right=of STAR] {$\ic{\COMM}$};
    \node (CIRC) [above=of COMM] {$\ic{\CIRC}$};

    \draw[hier, bend right] (MON) to node[pos=.45, edgeLabel]{\small\cite{Koedding.Truthe.2024}} (STAR);

    \draw[hier] (FIN) to node[pos=.45, edgeLabel]{\small\cite{Dassow_Manea_Truthe.2012}} (NIL);

    \draw[hier] (MON) to node[pos=.45, edgeLabel]{\small\cite{Dassow_Manea_Truthe.2012}} (NIL);

    \draw[hier] (MON) to node[pos=.45, edgeLabel]{\small\cite{Dassow_Manea_Truthe.2012}} (COMB);

    \draw[hier] (COMB) to node[pos=.45, edgeLabel]{\small\cite{Dassow_Manea_Truthe.2012}} (DEF);

    \draw[hier] (COMB) to node[pos=.45, edgeLabel]{\small\cite{Koedding.Truthe.2024}} (SYDEF);

    \draw[hier] (NIL) to node[pos=.45, edgeLabel]{\small\cite{Dassow_Manea_Truthe.2012}} (DEF);

    \draw[hier] (DEF) to node[pos=.45, edgeLabel]{\small\cite{Truthe.2014}} (ORD);

    \draw[hier, dashed] (ORD) to (NC);
    
    \draw[hier] (NC) to node[pos=.45, edgeLabel]{\small\cite{Truthe.2021}} (PS);
    
    \draw[hier] (PS) to node[pos=.45, edgeLabel]{\small\cite{Truthe.2021}} (REG);
    
    \draw[hier,bend right=22] (STAR) to node[pos=.45, edgeLabel]{\small\cite{Koedding.Truthe.2024}} (REG);

    \draw[hier] (MON) to 
    (INF);

    \draw[hier] (INF) to 
    (SUF);

    \draw[hier] (INF) to 
    (PRE);

    \draw[hier, bend right] (SYDEF) to node[pos=.45, edgeLabel]{\small\cite{Koedding.Truthe.2024}} (PS);

    \draw[hier, bend right=20] (SUF) to node[pos=.45, edgeLabel]{\small\cite{Truthe.2021}} (PS);

    \draw[hier, bend right=20] (PRE) to 
    (PS);

    \draw[hier, bend right] (MON) to node[pos=.45, edgeLabel]{\small\cite{Dassow_Manea_Truthe.2012}} (COMM);

    \draw[hier] (COMM) to node[pos=.45, edgeLabel]{\small\cite{Dassow_Manea_Truthe.2012}} (CIRC);

    \draw[hier, bend right = 10] (CIRC) to node[pos=.45, edgeLabel]{\small\cite{Dassow_Manea_Truthe.2012}} (REG);
    
  \end{tikzpicture}
  }
  \caption{Resulting hierarchy of language families by internal contextual grammars with special selection languages}
  \label{fig:ic_erg}
  \end{figure}

We now present some languages which 
serve 
as witness languages for proper inclusions or
incomparabilities.

\begin{lemma}\label{ic:suf_o_pre}
    Let $G = (\{a,b,c,d\}, \{\{ab,b,\lambda\} \to (c,d)\}, \{aab\})$ be a contextual grammar.
    Then, the language $L = L(G)$ is in $\ic{\SUF} \setminus \ic{\PRE}$.
\end{lemma}
\begin{proof}
    Since the selection language of $G$ is suffix-closed, the language $L$ is in $\ic{\SUF}$.
    Suppose that the language $L$ is also generated by a contextual grammar 
    $G' = (\{a,b,c,d\}, \cS', B')$
    where all selection languages $S\in\cS'$ are prefix-closed.
  
    Let us consider a word $w=ac^nabd^n \in L$ for some natural number $n \geq \ell(G')$. Due to the 
    choice of~$n$, the word~$w$ is derived in one step from some word $z_1z_2z_3\in L$ for three 
    words $z_i\in V^*$ with $1\leq i\leq 3$ by using a selection component $(S,C)\in\cS$ with $z_2\in S$ 
    and a context $(u,v)\in C$: $z_1z_2z_3\Lra z_1uz_2vz_3=w$. By the structure of the language $L$, 
    it holds $(u,v) = (c^m, d^m)$ for a natural number $m$ with $1 \leq m < n$ and~$z_2=c^pabd^q$ 
    for two natural numbers $p$ and $q$ with $m+p \leq n$ and $q + m \leq n$. Since $S$ is assumed 
    to be prefix-closed, the word $c^pa$ is in $S$, too. Therefore, we can apply the context $(u,v)$ 
    to the subword~$c^pa$ of~$w$. Hence, we can derive $ac^nabd^n \in L$ to $ac^{n+m}ad^mbd^n \notin L$. 
    This contradiction proves that~$L\notin\ic\PRE$.
\end{proof}

\begin{lemma}\label{ic:pre_o_suf}
    Let $G = (\{a,b,c,d\}, \{\{ab,a,\lambda\} \to (c,d)\}, \{abb\})$ be a contextual grammar.
    Then, the language $L = L(G)$ is in $\ic{\PRE} \setminus \ic{\SUF}$.
\end{lemma}
\begin{proof}
The argumentation is symmetrical to the previous proof.
\end{proof}

\begin{lemma}\label{ic:inf_o_nc}
    Let $V=\sets{a,b,c,d,e,f,g,h}$ be an alphabet, 
$G=(V,\sets{S_1\to C_1, S_2\to C_2},\sets{cd})$
be a contextual grammar with
\[S_1 = \Inf(\sets{a,b}^*\sets{cd}),\
C_1 = \sets{(aab,gh)},\ 
S_2 = \Inf(\sets{a}\sets{bb}^+\sets{c}),\
C_2 = \sets{(e,f)},\]
and $L=L(G)$ be its generated language.
Then, 
$L\in \cIC(\INF)\setminus\cIC(\NC).$
\end{lemma}
\begin{proof}
All selection languages are infix-closed, therefore, $L\in\cIC(\INF)$.
With non-counting selection languages, it could not be ensured that the number of letters $b$ between $e$ and $f$ in a word is even.
\end{proof}

\begin{lemma}\label{ic:inf_o_sydef}
Let $V=\Sets{a,b,c}$ be an alphabet,
$G=(V,\sets{\Inf(\sets{abca})\to (b,c)},\sets{abcaaabca})$
be a contextual grammar,
and $L=L(G)$ be its generated language.
Then, 
$L\in \cIC(\INF) \setminus \cIC(\SYDEF).$
\end{lemma}
\begin{proof}
The selection language of the given contextual grammar is infix-closed; 
therefore, $L\in\cIC(\INF)$.
With symmetric definite selection languages, letters $b$ could be produced in the beginning of a word whereas the corresponding letters $c$ are produced at the very end of a word which is a contradiction.
\end{proof}

\begin{lemma}\label{ic:inf_o_circ}
    Let $V=\{a, b, c, d\}$ be an alphabet,
$G=(V,\sets{\sets{a b, a, b, \lambda}\to (c, d)},\sets{aab, ba})$
be a contextual grammar, and $L=L(G)$ be the language generated. Then,
$L \in \ic{\INF} \setminus \ic{\CIRC}.$
\end{lemma}
\begin{proof}
    The selection language of 
    $G$ is infix-closed. Hence, we have $L \in \ic{\INF}$.
    In \cite[Lemma 17]{DasManTru12b}, it was proved that the language $L$ is not in $\ic{\CIRC}$.
\end{proof}

\begin{lemma}\label{ic:fin_comb_o_pre}
    Let $L = \set{c^nac^mbc^{n+m}}{n\geq 0,m \geq 0}$. Then, $L \in (\ic{\FIN} \cap \ic{\COMB}) \setminus \ic{\PRE}$.
\end{lemma}
\begin{proof}
    The relations $L \in \ic{\FIN}$ and $L \in \ic{\COMB}$ were proved in \cite[Lemma 15]{DasManTru12b}. The relation~$L \notin \ic{\PRE}$ can be proved in the same way as $L \notin \ic{\SUF}$ is proved in \cite[Lemma 15]{DasManTru12b}.
\end{proof}

\begin{lemma}\label{ic:inf_o_star}
    Let $V = \sets{a,b,c,d}$ be an alphabet, 
    $G = (V, \{S_1 \to C_1, S_2 \to C_2\}, \{baab\})$
    be a contextual grammar with 
    $S_1=\sets{a,\lambda},\ 
    C_1=\sets{(c,d)},\ 
    S_2=\sets{b,\lambda},\ 
    C_2=\sets{(d,c)},$ 
    and $L = L(G)$ its generated language. Then, 
    $L \in \ic{\INF} \setminus \ic{\STAR}.$
\end{lemma}
\begin{proof}
    Since both selection languages of $G$ are infix-closed, we have $L \in \ic{\INF}$.
    In 
    \cite{Koedding.Truthe.2025}, it is proved that~$L \notin \ic{\STAR}$.
\end{proof}

With the languages from the previous lemmas, the inclusion relations and incomparabilities depicted in Figure~\ref{fig:ic_erg} can be shown.

\begin{theorem}[Resulting hierarchy for $\cIC$]\label{theorem:neue_hierarchie_IC}
The inclusion relations presented in Figure \ref{fig:ic_erg} hold. An arrow from an entry $X$ to
an entry~$Y$ depicts the proper inclusion $X \subset Y$; if two families are not connected by a directed
path, they are incomparable.
\end{theorem}

Please note that with the result in Theorem~\ref{theorem:neue_hierarchie_IC}, we have answered the open question whether~$\ic\SUF$ is incomparable to $\ic\NC$ or $\ic\ORD$ or whether it is a subset of one of the families $\ic\NC$ or~$\ic\ORD$ raised already several years ago in \cite{Truthe.2017}.

\section{Conclusion and future work}

In this paper, we have extended the previous hierarchies of subregular language families,
of families generated by external contextual grammars with selection in certain subregular language families,
and of families generated by internal contextual grammars with selection in such language families.

Various other subregular language families have also been investigated in the past 
(for instance, in~\cite{Bordihn_Holzer_Kutrib.2009, Han_Salomaa.2009, Olejar_Szabari.2023}). 
Future research will be on extending and unifying current hierarchies of
subregular language families 
(presented, for instance, in \cite{Dassow_Truthe.2023,Truthe.2021})
by additional families and to use them as control in contextual grammars. 

The extension of the hierarchy with other families of definite-like languages (for instance,
ultimate definite, central definite, non-inital definite) has also already begun. 
Furthermore, it is also planned to unify the hierarchy of subregular language families, 
extended by the mentioned language families, with the hierarchies of the language families 
generated by contextual grammars defined by their limited resources.

The research can be also extended to other mechanisms like tree-controlled grammars
or networks of evolutionary processors. Another possibility would be to check to what extent the different language classes are closed under different operations. 

\bibliographystyle{eptcs}
\bibliography{ncma25}
\newpage
\end{document}